\documentclass[conference,10pt]{IEEEtran}

\usepackage{mathtools}

\usepackage{amsmath}
\usepackage{amsthm,amssymb}

\newtheorem{theo}{Theorem}

\newtheorem{lemma}[theo]{Lemma}

\usepackage{comment}
\usepackage{graphicx}
\usepackage{adjustbox}

\DeclareMathOperator*{\argmin}{arg\,min}
\begin{document}

\title{Non-Asymptotic Performance Analysis of Size-Based Routing Policies}

\author{\IEEEauthorblockN{Eitan Bachmat}
\IEEEauthorblockA{Department of Computer Science\\Ben-Gurion University\\Beer-Sheva, Israel, 84105. \\
ebachmat@cs.bgu.ac.il}
\and
\IEEEauthorblockN{Josu Doncel}
\IEEEauthorblockA{Department of Mathematics \\ University of the Basque Country, UPV/EHU\\
Leioa, Spain. 48940\\
josu.doncel@ehu.eus}
}

\maketitle
\thispagestyle{plain}
\pagestyle{plain}
\begin{abstract}
We investigate the performance of two size-based routing policies: the Size Interval Task Assignment (SITA) and 
 Task Assignment based on Guessing Size (TAGS).
We consider a system with two servers and Bounded Pareto distributed job sizes with tail parameter 1 where
the difference between the size of the largest and the smallest job is finite. 
We show that the ratio between the mean waiting 
time of TAGS over the mean waiting time of SITA is unbounded when the largest job size is large and the arrival rate times the largest job size is less than one.
We provide numerical experiments that show that our theoretical findings extend to Bounded Pareto distributed job sizes 
with tail parameter different to $1$.
\end{abstract}


\section{Introduction}
\begin{figure*}[t!]
\begin{minipage}[b]{0.45\linewidth}
\centering
\includegraphics[width=\columnwidth,clip=true,trim=0pt 480pt 0pt 70pt]{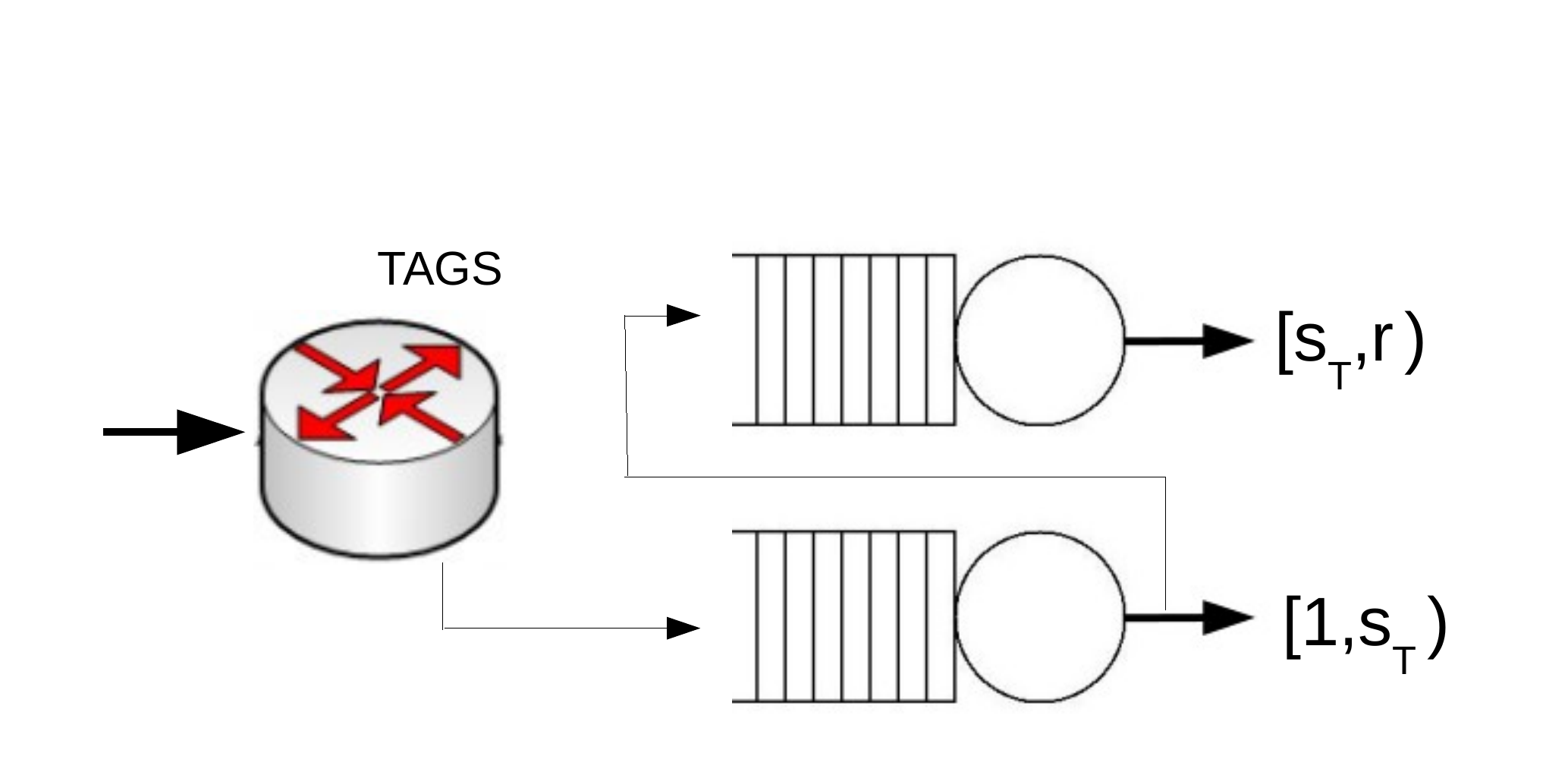}
\caption{A system operating under the TAGS policy}
\label{fig:tags-example}
\end{minipage}
\hspace{0.5cm}
\begin{minipage}[b]{0.45\linewidth}
\centering
\includegraphics[width=\columnwidth,clip=true,trim=0pt 480pt 0pt 70pt]{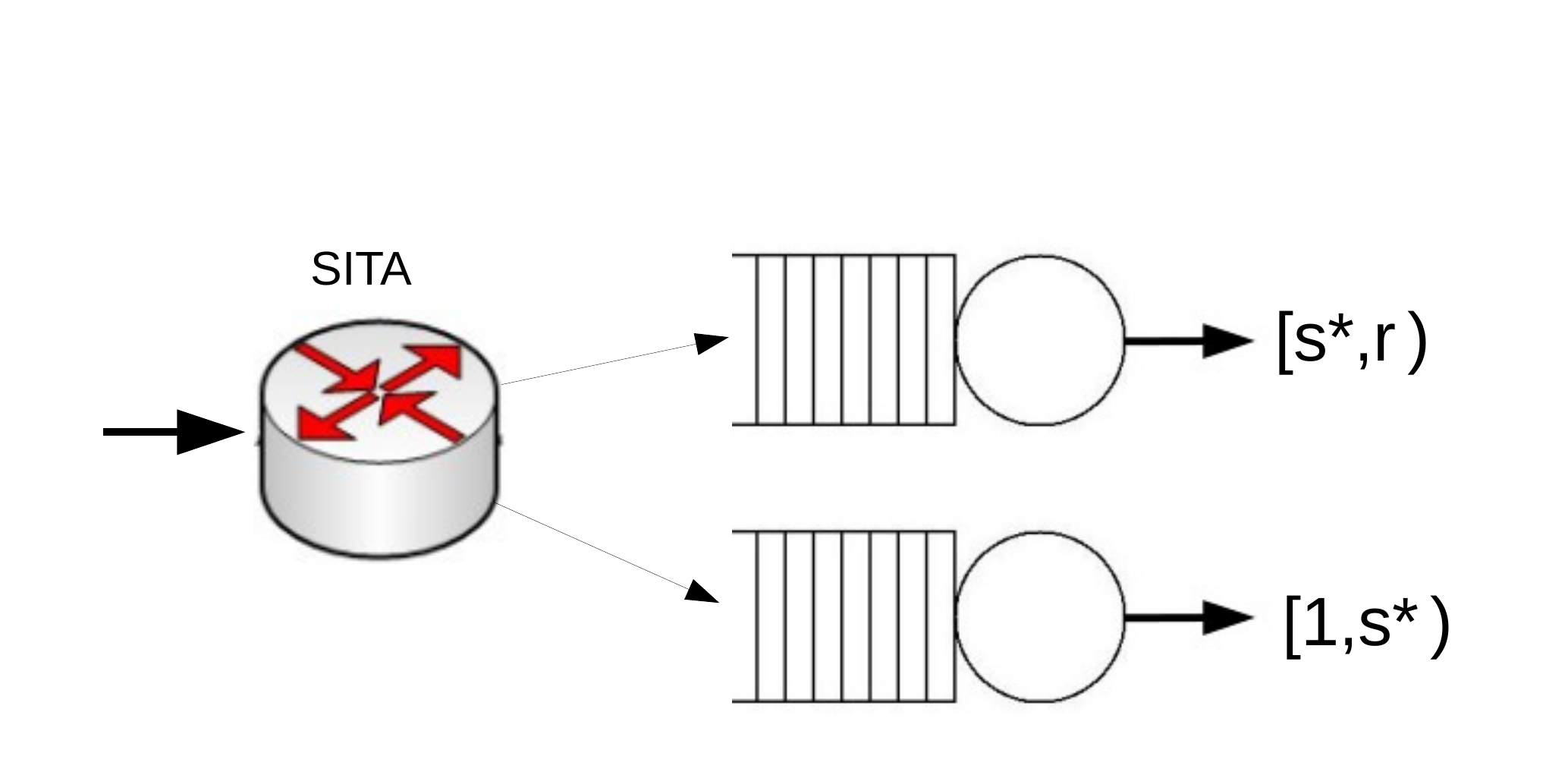}
\caption{A system operating under the SITA policy}
\label{fig:sita-example}
\end{minipage}
\end{figure*}

Many existing routing policies to parallel queues are included in the SQ(d) framework, where for each incoming job, $d  \geq 2$ 
servers are picked uniformly at random to observe their states and the job is routed to the server in the best state among the observed ones. It is known
that the performance of this kind of systems is very good (see \cite{foley2001join,weber1978optimal,richa2001power,winston1977optimality}), however 
the author in \cite{whitt1986deciding} showed that when the variability of the job size distribution is high this
family of policies are not optimal. This is, in fact, the regime where the size-based routing policies outperform 
the routing policies that belong the SQ(d) family of policies.

In this work, we study two size-based policies: the Size Interval Task Assignment (SITA) \cite{harchol1999choosing} and
the Task Assignment with Guessing Size (TAGS) \cite{Ha}. In the former policy, short jobs and long jobs are executed in different servers
and, therefore, it is assumed that the size of incoming tasks is known. In the latter policy, all the jobs are executed in one server
and, if a job does not end its service before a given deadline, it is stopped and enqueued in the other server, where it starts service from scratch. From the above definitions, it follows clearly that
the mean waiting time of jobs in the SITA system is smaller than in the TAGS system. Recently,
the authors in \cite{Bachmat2019PerformanceAS,BACHMAT2020102122} show for Bounded Pareto distributed job sizes that, 
when the difference between the smallest and the largest job
size tends to infinity and the total load of the system is less than one, the ratio of the mean waiting time of TAGS over the mean waiting time of SITA is upper bounded by 2. 

We consider a system with two servers and we compare the mean waiting time of both systems 
in a non-asymptotic regime where
the difference between the size of the largest and the smallest job is finite. We show that the ratio of the mean 
waiting time of the TAGS system over that of the SITA system is unbounded for the Bounded Pareto distribution with
tail parameter $1$ when the largest job size is large and the arrival rate times the largest job size is less than one.
Our numerical experiments show that the performance 
ratio of both systems is unbounded other values of the tail parameter of the Bounded Pareto distribution and it can be large
even if the arrival rate is not small.

\section{Model Description}
\label{sec:model}

\subsection{Notation}
We consider a system with two servers with equal capacity. We assume that the servers are 
FCFS queues. Arriving jobs follow a Poisson distribution with rate $\lambda$. 
The size of the jobs is given by a sequence of i.i.d. random variables denoted by $X$. We denote by 
$F(s)=\mathbb P[X<s]$ the cumulative distribution function of the job size distribution. 
We  assume $F(\cdot)$ to  be differentiable  and  we write $f(s)  =\frac{dF(s)}{ds}$.
We assume that the size of the smallest job is one and the size of the 
largest job is $r>1$. The load of the system is denoted by $\rho$ and to ensure stability we assume that $\rho<1$.

\subsection{The TAGS System}
We focus on the TAGS routing (see Figure~\ref{fig:tags-example}). Let $s\in[1,r].$ In 
this policy, all incoming jobs are sent to Server $1$. If a job has been served before $s$ 
units of time in Server 1, it leaves the system; 
otherwise, when the execution time equals $s$, it is stopped and sent to the end of the queue of Server 2, 
where the execution starts from scratch. Jobs that are executed in Server 2 are always executed until completion.

For a given value $s$, we denote by $W^{TAGS}(s)$ the random variable of 
the waiting time of incoming jobs under TAGS policy. We define as $s_T$ the value of $s$ such that the mean waiting time 
of the TAGS system is minimized, i.e., 
$$
s_T = \argmin_s \mathbb E[W^{TAGS}(s)].
$$

From the above description, it follows that
\begin{align}
\mathbb E[W^{TAGS}(s_T)]=&\mathbb E[W_1^{TAGS}(s_T)]\nonumber\\&+
\left(\int_{s_T}^rf(x)dx\right)s_T\nonumber\\&+\left(\int_{s_T}^rf(x)dx\right)\mathbb E[W_2^{TAGS}(s_T)],
\label{eq:tags}
\end{align}
where $\mathbb E[W_i^{TAGS}(s_T)]$ is the mean waiting time of jobs in Server $i$ when the threshold value is $s_T$.

\subsection{The SITA System}
We focus on the SITA routing  (see Figure~\ref{fig:sita-example}). This policy assumes that the size of incoming jobs is known \cite{DAA19,anselmi2019asymptotically}. Let $s\in[1,r].$ 
Under the SITA policy with cutoff $s$, jobs whose size is smaller than $s$ are sent to Server 1, whereas
jobs whose size is larger than $s$ to Server 2. The random variable of the waiting time of incoming jobs under the SITA
policy with cutoff $s$ is denoted by $W^{SITA}(s)$. We define as $s^*$ the value of $s$ such that the mean waiting time
of the SITA system is minimized, i.e., 
$$
s^* = \argmin_s \mathbb E[W^{SITA}(s)].
$$

The mean waiting time of jobs under the SITA policy with cutoff $s^*$ is given by

\begin{align}
\mathbb E[W^{SITA}(s^*)]&=\left(\int_1^{s^*}f(x)dx\right)\mathbb E[W_1^{SITA}(s^*)]\nonumber\\&+
\left(\int_{s^*}^rf(x)dx\right)\mathbb E[W_2^{SITA}(s^*)],
\label{eq:sita}
\end{align}
where $\mathbb E[W_i^{SITA}(s^*)]$ is the mean waiting time of jobs in Server $i$.

\subsection{The Bounded Pareto Distribution}

We now present the Bounded Pareto distribution. 
A distribution $X$ is said to be Bounded Pareto with parameters $1$, $r$ and $\alpha$ if its density 
has the following form: if $1\leq x \leq r$, then 
$$
f(x)=\frac{\alpha x^{-\alpha-1}}{(1-(1/r)^\alpha)},
$$
and $f(x)=0$ otherwise. Besides, the cumulative distribution function of the Bounded
Pareto distribution is given by the following expression:
\[ F(x)=\begin{cases} 
      0, & x\leq 1, \\
      \frac{1-(1/x)^{\alpha}}{1-(1/r)^{\alpha}}, & 1\leq x\leq r, \\
      1, & x \geq r.
   \end{cases}
\]
Besides, when $\alpha\neq 1$, we have that the mean of the distribution is given by
$
\frac{\alpha }{\alpha-1}\frac{1-(1/r)^{\alpha -1}}{1-(1/r)^{\alpha }}
$
whereas when $\alpha =1$
$
\frac{\ln (r)}{1-\frac{1}{r}}.
$
The Bounded Pareto distribution with large range and $0<\alpha<2$ is know to be a good model for high
variance job size distributions \cite{HD1997}.
We also note that the Bounded Pareto distribution with $\alpha =-1$ coincides with the uniform distribution on the interval $[1,r]$. Besides, the most favorable distribution of SITA and TAGS is given when $\alpha=1$ \cite{Bachmat2019PerformanceAS}. This is, indeed, the case we study in the next section.
\section{Bounded Pareto Distribution with $\alpha = 1$}
\label{sec:light}
We consider the Bounded Pareto distribution with $\alpha = 1$ and
we aim to compare the mean waiting time of a system with two queues operating under the 
TAGS routing with that of a system with two queues operating under the SITA routing. In the following
result, we provide a lower-bound of the performance of the TAGS system.
\begin{lemma}
For the Bounded Pareto distribution with $\alpha = 1$, if $\lambda r<1$,
$$
\mathbb E[W^{TAGS}(s_T)]> \lambda r.
$$
\label{lem:tags-lowerbound}
\end{lemma}
\begin{proof}
We aim to study the optimal performance of the TAGS system with two queues. It follows from
\eqref{eq:tags} that
\begin{align*}
\mathbb E[W^{TAGS}(s_T)]\geq &\mathbb E[W_1^{TAGS}(s_T)]\\&+
\left(\int_{s_T}^rf(x)dx\right)s_T.
\end{align*}
For the Bounded Pareto distribution with $\alpha=1$, we have that
\begin{align*}
&\mathbb E[W_1^{TAGS}(s_T)]+
\left(\int_{s_T}^rf(x)dx\right)s_T=\\&\frac{\lambda(s_T-1)}{2(1-\rho)(1-\tfrac{1}{r})}+\frac{1-\tfrac{s_T}{r}}{1-\tfrac{1}{r}}\geq\\
&\frac{\lambda(s_T-1)}{2(1-\tfrac{1}{r})}+\frac{1-\tfrac{s_T}{r}}{1-\tfrac{1}{r}}=
\frac{s_T(\lambda-\tfrac{1}{r})+2-2\lambda}{2(1-\tfrac{1}{r})}.
\end{align*}
If $\lambda r<1$, then 
$\frac{x(\lambda-\tfrac{1}{r})+2-2\lambda}{2(1-\tfrac{1}{r})}$ 
decreases with $x$. Thus,
 \begin{align*}
\frac{s_T(\lambda-\tfrac{1}{r})+2-2\lambda}{2(1-\tfrac{1}{r})}&\geq  \frac{r(\lambda-\tfrac{1}{r})+2-2\lambda}{2(1-\tfrac{1}{r})}\\
&=\frac{\lambda (r-2)+1}{2(1-\tfrac{1}{r})}>\frac{\lambda (r-2)+ \lambda r}{2(1-\tfrac{1}{r})}=\lambda r.
\end{align*}
\end{proof}
It is important to note that, unlike in the previous work \cite{Bachmat2019PerformanceAS}, we do not need to 
assume Poisson arrivals to all the servers in the above result. 
In the following result, we characterize the mean waiting time of the SITA system.
\begin{lemma}
For the Bounded Pareto distribution with $\alpha=1$, when $\lambda r<1$ and $r$ is large,
$$
\mathbb E[W^{SITA}(s^*)]\leq \frac{\lambda(\sqrt{r}-1)^2}{\sqrt{r}(1- \tfrac{1}{r})^2}.
$$
\label{lem:sita-value}
\end{lemma}
\begin{proof}
We first note that the load of each server is upper bounded by $\rho$. Therefore,
for the Bounded Pareto distribution with $\alpha=1$, we have that
$$
\mathbb E[W_1^{SITA}(s^*)]\leq \frac{ \lambda(s^*-1)}{2(1-\rho)(1-(1/r))}
$$
and
$$
\mathbb E[W_2^{SITA}(s^*)]\leq \frac{ \lambda(r-s^*)}{2(1-\rho)(1-(1/r))},
$$
where $\rho=\lambda \ln(r)$. We now note that, 
$$
\lambda r<1 \iff  \lambda \ln(r)<\frac{ln(r)}{r}.
$$
Since when $r$ is large, $\frac{\ln(r)}{r}$ tends to zero and $\rho=\lambda \ln(r)$, it follows that $\rho$ tends to zero when $\lambda r<1$.
As a result, 
$$
\mathbb E[W_1^{SITA}(s^*)]\leq \frac{ \lambda(s^*-1)}{2(1-(1/r))}
$$
and
$$
\mathbb E[W_2^{SITA}(s^*)]\leq \frac{ \lambda(r-s^*)}{2(1-(1/r))}.
$$
We know from Section 3.2.4 of \cite{HV10} that, for the Bounded Pareto distribution with $\alpha=1$, 
$s^*$ balances the load of both servers and, therefore, it can be obtained as follows:
$$
\int_1^{s^*}f(x)dx=\int_{s^*}^rf(x)dx \iff s^*=\sqrt{r}.
$$
As a result,
$$
\mathbb E[W_1^{SITA}(s^*)]\leq \frac{ \lambda(\sqrt{r}-1)}{2(1-(1/r))}
$$
and
$$
\mathbb E[W_2^{SITA}(s^*)]\leq \frac{ \lambda(r-\sqrt{r})}{2(1-(1/r))}.
$$

Therefore, from \eqref{eq:sita}, it follows that
\begin{align*}
\mathbb E[W^{SITA}(s^*)]&\leq \left(\int_1^{s^*}f(x)dx\right)\frac{ \lambda(\sqrt{r}-1)}{2(1-(1/r))}\nonumber\\&+
\left(\int_{s^*}^rf(x)dx\right)\frac{ \lambda(r-\sqrt{r})}{2(1-(1/r))}\\
&=\left(\int_1^{\sqrt{r}}f(x)dx\right)\frac{ \lambda(\sqrt{r}-1)}{2(1-(1/r))}\nonumber\\&+
\left(\int_{\sqrt{r}}^rf(x)dx\right)\frac{ \lambda(r-\sqrt{r})}{2(1-(1/r))}\\
&=\frac{ \lambda(1-\tfrac{1}{\sqrt{r}})(\sqrt{r}-1)}{2(1-(1/r))^2}\\&+
\frac{ \lambda(\tfrac{1}{\sqrt{r}}-\tfrac{1}{r})(r-\sqrt{r})}{2(1-(1/r))^2}\\&=
\frac{\lambda(\tfrac{1}{\sqrt{r}}-\tfrac{1}{r})(r-\sqrt{r})}{(1-(1/r))^2}\\&=
\frac{\lambda(\sqrt{r}-1)^2}{\sqrt{r}(1-(1/r))^2}.
\end{align*}
\end{proof}
From the above lemmas, we have that when $\lambda r<1$ and $r$ is large, the ratio between the mean waiting time of TAGS over the mean waiting time of SITA is lower bounded by
$$
 \frac{\lambda r}{\frac{\lambda(\sqrt{r}-1)^2}{\sqrt{r}(1-(1/r))^2}}=\frac{ (\sqrt{r}+1)^2}{\sqrt{r}},
$$
where the equality is obtained by simplifying the derived expression. Interestingly, the last expression does not depend on 
$\lambda$. Besides, we now show that it is increasing with $r$.

\begin{lemma}
The function $\frac{ (\sqrt{r}+1)^2}{\sqrt{r}}$ is increasing with $r$.
\label{prop:main-result}
 \end{lemma}

\begin{proof}
We aim to show that $\frac{ (\sqrt{r}+1)^2}{\sqrt{r}}$ is increasing with $r$, which is true if and only if
$$
((\sqrt{r}+1)^2)' \sqrt{r}-(\sqrt{r})' (\sqrt{r}+1)^2>0 \iff
$$
$$
\frac{2(\sqrt{r}+1)}{2\sqrt{r}}\sqrt{r}-\frac{1}{2\sqrt{r}} (\sqrt{r}+1)^2>0 \iff
$$
$$
\sqrt{r}+1-\frac{1}{2\sqrt{r}} (\sqrt{r}+1)^2>0 \iff
$$
$$
1-\frac{1}{2\sqrt{r}} (\sqrt{r}+1)>0 \iff
$$
$$ 
2\sqrt{r}- (\sqrt{r}+1)>0 \iff \sqrt{r}- 1>0
$$
%
\end{proof}

From the above results and using that 
$\frac{ (\sqrt{r}+1)^2}{\sqrt{r}}$ 
tends to infinity when $r\to \infty$, it follows that
when $\lambda r<1$, the ratio between the mean waiting time of the TAGS system and the mean waiting time of the SITA 
system is lower bounded by a function that is unbounded and the following result follows.

\begin{theo}
The ratio of the mean waiting time of the TAGS system and the mean waiting time of the SITA 
system is unbonded.
\end{theo}

\section{Numerical Experiments}
\label{sec:numerics}
We present our numerical work, which focuses on the ratio between the mean waiting time of the TAGS system and the mean waiting time of the SITA 
system for the Bounded Pareto distribution. We investigate the evolution of 
this ratio when we vary $r$ from $10$ to $1000$ for
different values of $\lambda$ and different values of $\alpha$.
In Figure~\ref{fig:sita-vs-tags-lambda}, we consider that $\alpha = 1$ and we depict the ratio between the mean waiting time of the TAGS system and the mean waiting time of the SITA 
system as a function
of $r$ for different values of $\lambda$. We observe that the maximum of this ratio is 4 when $\lambda=0.05$, 8 when $\lambda=0.01$,
$11$ when $\lambda=0.005$ and when $\lambda=0.001$ the performance ratio is increasing with $r$ for the considered values 
of $r$. This illustration shows that the maximum  over $r$ of the performance ratio under consideration increases 
when we decrease $\lambda$ and also that it can be large even if the value of the arrival rate is not very small.
In Figure~\ref{fig:sita-vs-tags-alpha}, we consider an arrival rate of $0.001$ and we study the evolution 
of the ratio between the mean waiting time of the TAGS system and the mean waiting time of the SITA 
system over $r$ for several values of $\alpha$ larger than one. 
We observe that the performance ratio under study is increasing with $r$ for the considered values of $r$. Similar results have
been obtained for this performance ratio for different values of $\alpha$ which are smaller than one, that is, 
ratio between the mean waiting time of TAGS and the mean waiting time of SITA is also increasing with $r$
when $\alpha$ is smaller than 1. We omit this illustration due to lack of space.
\begin{figure}
\centering
\includegraphics[width=0.95\columnwidth,clip=true,trim=10pt 200pt 10pt 210pt]{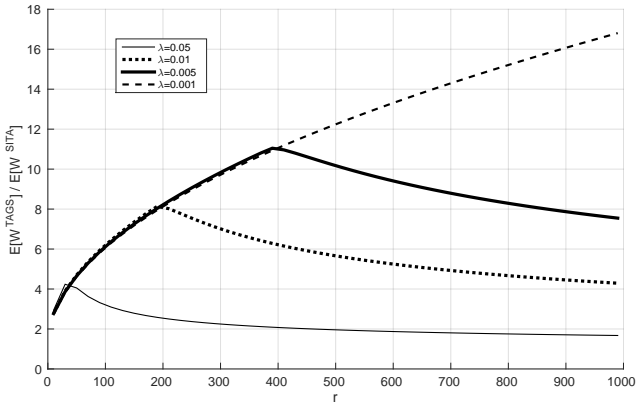}
\caption{The ratio $\mathbb E[W^{TAGS}(s_T)] / \mathbb E[W^{SITA}(s_T)]$ as a function of $r$ for different values of $\lambda$.}
\label{fig:sita-vs-tags-lambda}
\end{figure}
\begin{figure}
\centering
\includegraphics[width=0.95\columnwidth,clip=true,trim=0pt 200pt 10pt 210pt]{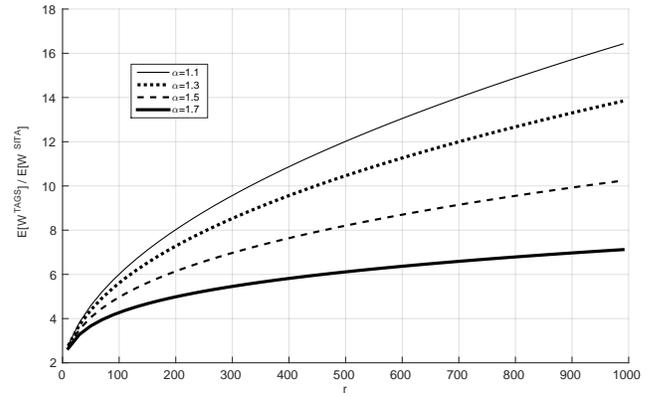}
\caption{The ratio $\mathbb E[W^{TAGS}(s_T)] / \mathbb E[W^{SITA}(s_T)]$ as a function of $r$ for different values of $\alpha$.}
\label{fig:sita-vs-tags-alpha}
\end{figure}
\section*{Acknowledgements}
 Josu Doncel has received funding from the Department of Education of the Basque Government through the Consolidated Research Group MATHMODE (IT1294-19), from theMarie Sklodowska-Curie grant agreement No 777778 and from from the Spanish Ministry of Science and Innovation with reference PID2019-108111RB-I00 (FEDER/AEI).\\
Eitan Bachmat's  work was supported by the German Science Foundation (DFG) through the grant, Airplane Boarding, (JA 2311/3-1). 
\bibliographystyle{IEEEtran}
\bibliography{IEEEabrv,tagsbib}

\begin{thebibliography}{10}
\providecommand{\url}[1]{#1}
\csname url@samestyle\endcsname
\providecommand{\newblock}{\relax}
\providecommand{\bibinfo}[2]{#2}
\providecommand{\BIBentrySTDinterwordspacing}{\spaceskip=0pt\relax}
\providecommand{\BIBentryALTinterwordstretchfactor}{4}
\providecommand{\BIBentryALTinterwordspacing}{\spaceskip=\fontdimen2\font plus
\BIBentryALTinterwordstretchfactor\fontdimen3\font minus
  \fontdimen4\font\relax}
\providecommand{\BIBforeignlanguage}[2]{{%
\expandafter\ifx\csname l@#1\endcsname\relax
\typeout{** WARNING: IEEEtran.bst: No hyphenation pattern has been}%
\typeout{** loaded for the language `#1'. Using the pattern for}%
\typeout{** the default language instead.}%
\else
\language=\csname l@#1\endcsname
\fi
#2}}
\providecommand{\BIBdecl}{\relax}
\BIBdecl

\bibitem{foley2001join}
R.~D. Foley, D.~R. McDonald \emph{et~al.}, ``Join the shortest queue: stability
  and exact asymptotics,'' \emph{The Annals of Applied Probability}, vol.~11,
  no.~3, pp. 569--607, 2001.

\bibitem{weber1978optimal}
R.~R. Weber, ``On the optimal assignment of customers to parallel servers,''
  \emph{Journal of Applied Probab.}, vol.~15, no.~2, pp. 406--413, 1978.

\bibitem{richa2001power}
A.~W. Richa, M.~Mitzenmacher, and R.~Sitaraman, ``The power of two random
  choices: A survey of techniques and results,'' \emph{Combinatorial
  Optimization}, vol.~9, pp. 255--304, 2001.

\bibitem{winston1977optimality}
W.~Winston, ``Optimality of the shortest line discipline,'' \emph{Journal of
  Applied Probability}, vol.~14, no.~1, pp. 181--189, 1977.

\bibitem{whitt1986deciding}
W.~Whitt, ``Deciding which queue to join: Some counterexamples,''
  \emph{Operations Research}, vol.~34, no.~1.

\bibitem{harchol1999choosing}
M.~Harchol-Balter, M.~E. Crovella, and C.~D. Murta, ``On choosing a task
  assignment policy for a distributed server system,'' \emph{Journal of
  Parallel and Distributed Computing}, vol.~59, no.~2, pp. 204--228, 1999.

\bibitem{Ha}
M.~Harchol-Balter, ``Task assignment with unknown duration,'' in
  \emph{Proceedings 20th IEEE International Conference on Distributed Computing
  Systems}.\hskip 1em plus 0.5em minus 0.4em\relax IEEE, 2000, pp. 214--224.

\bibitem{Bachmat2019PerformanceAS}
E.~Bachmat, J.~Doncel, and H.~Sarfati, ``Performance and stability analysis of
  the task assignment based on guessing size routing policy,'' \emph{IEEE 27th
  International Symposium on Modeling, Analysis, and Simulation of Computer and
  Telecommunication Systems}, pp. 1--13, 2019.

\bibitem{BACHMAT2020102122}
------, ``Analysis of the task assignment based on guessing size policy,''
  \emph{Performance Evaluation}, vol. 142, p. 102122, 2020.

\bibitem{DAA19}
J.~{Doncel}, S.~{Aalto}, and U.~{Ayesta}, ``Performance degradation in
  parallel-server systems,'' \emph{IEEE/ACM Trans. on Netw.}, vol.~27, no.~2,
  pp. 875--888, April 2019.

\bibitem{anselmi2019asymptotically}
J.~Anselmi and J.~Doncel, ``Asymptotically optimal size-interval task
  assignments,'' \emph{IEEE Transactions on Parallel and Distributed Systems},
  vol.~30, no.~11, pp. 2422--2433, 2019.

\bibitem{HD1997}
M.~Harchol-Balter and A.~B. Downey, ``Exploiting process lifetime distributions
  for dynamic load balancing,'' \emph{ACM Trans. Comput. Syst.}, vol.~15,
  no.~3, pp. 253--285, Aug. 1997.

\bibitem{HV10}
M.~{Harchol-Balter} and R.~{Vesilo}, ``To balance or unbalance load in
  size-interval task allocation,'' \emph{Probability in the Engineering and
  Informational Sciences}, vol.~24, no.~2, pp. 219--244, 2010.

\end{thebibliography}

\end{document}